\documentclass{birkjour}
\usepackage[all]{xy}

 \newtheorem{thm}{Theorem}[section]

 \newtheorem{prop}[thm]{Proposition}
 \theoremstyle{definition}
 
 \theoremstyle{remark}

 \numberwithin{equation}{section}

\begin{document}

%
%
%
%
%
%
%
%
%

\title[On covariant Poisson brackets in field theory]
 {On covariant Poisson brackets in field theory}

\author[A. A. Sharapov]{A. A. Sharapov}

\address{%
Tomsk State University\\
Physics Faculty\\
Lenin ave. 36\\
Tomsk 634050,
Russia}

\email{sharapov@phys.tsu.ru}

\thanks{This work was partially supported by the Tomsk State University Academic  D.I. Mendeleev Fund Program, the RFBR grant
13-02-00551 and  the Dynasty Foundation.}



\date{January 1, 2014}

\begin{abstract}
A general approach is proposed to constructing covariant Poisson brackets in the space of histories of a classical field-theoretical model. The approach is based on the concept of Lagrange anchor, which was originally developed as a tool for path-integral quantization of Lagrangian and non-Lagrangian dynamics. The proposed covariant Poisson brackets generalize the Peierls' bracket construction known in the Lagrangian field theory.
\end{abstract}

\maketitle

\section{Introduction}

The least action principle provides the foundation for classical mechanics and field theory.
A distinguishing feature of  the Lagrangian equations of motion among other differential equations is that their  solution space carries a natural symplectic structure, making it into a phase space. The physical observables, being identified with the smooth function(al)s on the phase space, are then endowed with the structure of a Poisson algebra. There are at least two different ways for describing this Poisson algebra. The first one is the standard Hamiltonian formalism, which requires an explicit splitting of space-time into space and time and introduction of canonical momenta. The main drawback of this approach is the lack of manifest covariance, which  causes some complications in applying it to relativistic field theory.  An alternative approach was proposed by Peierls in his seminal 1952 paper \cite{P}. In that paper he invented what is now known as the Peierls brackets on the covariant phase space. In contrast to the usual (non-covariant) Hamiltonian formalism, where the phase space is identified with the space of initial data, the covariant phase space is the functional space consisting of all the trajectories obeying the Lagrangian equations of motion. Peierls' paper opened up the way for constructing a fully relativistic theory of quantum fields \cite{DWbook}.   For more recent discussions of the Peierls brackets, on different levels of rigor,  we refer the reader to \cite{M, FR, Kh}.

 In this paper, we explain how to extend the concept of covariant phase-space to the most  general (i.e., not necessarily Lagrangian) theories. Our approach is based on the notion of Lagrange anchor, which was originally proposed in \cite{KazLS} as a tool for path-integral quantization of Lagrangian and non-Lagrangian theories. In most cases the existence of a Lagrange anchor appears to be  less restrictive condition for the classical dynamics than the existence of an action functional. Furthermore, one and the same system of equations may admit a variety of different Lagrange anchors leading to inequivalent quantizations. In the next sections, we will show that any Lagrange anchor gives rise to a Poisson structure in the space of solutions to the classical equations of motion. The corresponding  Poisson brackets are fully covariant and reduce to the Peierls brackets in the case of Lagrangian theories endowed with the canonical Lagrange anchor. It is pertinent to note that for the mechanical systems described by ordinary differential equations in normal form, a relationship between the Lagrange anchors and Poisson brackets has been already established in \cite{KLS0}.

Our exposition is mostly focused on the algebraic and geometric aspects of the construction, while more subtle functional analytical details are either ignored or treated in a formal way. These details, however, are not specific to our problem  and can be studied, in principle, along the same lines as in the case of the conventional Peierls' brackets.

\section{Classical gauge systems}

\subsection{Kinematics} In modern language the  classical fields are just the sections of a locally trivial fiber bundle $B\rightarrow M$ over the space-time manifold $M$. The typical fiber $F$ of $B$ is called the target space of fields. In case the bundle is trivial, i.e., $B=M\times F$, the fields are merely  the mappings from $M$ to $F$. In each trivializing coordinate chart $U\subset M$ a  field $\varphi: M\rightarrow B$ is described by a collection of functions $\varphi^i(x)$, where $x\in U$ and $\varphi^i$ are local coordinates in $F$. These functions are often called the components of the field $\varphi$.

Formally, one can think of $\Gamma(B)$ -- the space of all field configurations -- as a smooth manifold $\mathcal{M}$ with the continuum infinity of dimensions and $\varphi^i(x)$ playing the role of local coordinates. In other words, the different local coordinates $\varphi^i(x)$ on $\mathcal{M}$ are labeled by the space-time point $x\in M$ and the discrete index $i$. To emphasize this interpretation of fields as coordinates on the infinite-dimensional manifold $\mathcal{M}$ we will include the space-time point $x$ into the discrete index $i$ and write $\varphi^i$ for $\varphi^i(x)$; in so doing, the summation over the ``superindex'' $i$ implies usual summation for its discrete part and integration over $M$ for $x$. In the physical literature this convention is known as DeWitt's condensed notation \cite{DWbook}.

Proceeding with the infinite-dimensional geometry above,  we identify  the ``smooth functions'' on the ``manifold'' $\mathcal{M}$ with the infinitely differentiable functionals of field $\varphi$. These functionals form a commutative algebra, which will be denoted by $\Phi$. If $\delta\varphi^i$ is an infinitesimal variation of field, then, according to the condensed notation, the corresponding variation of a functional $S\in \Phi$ can be written in the form
\begin{equation}\label{dS}
\delta S=S,_i\delta\varphi^i\,,
\end{equation}
where the comma denotes the functional derivative.

The concepts of vector fields, differential forms and exterior differentiation on $\mathcal{M}$ are naturally introduced through the functional derivatives, see e.g. \cite{KLS1}. In particular,  the variations $\delta \varphi^i$ span the space of 1-forms and the functional derivatives
$\delta/\delta\varphi^i$ define a basis in the tangent space  $T_\varphi \mathcal{M}$. So, we can speak of the tangent and cotangent bundles of $\mathcal{M}$.

The tangent and cotangent bundles are not the only vector bundles that can be defined over $\mathcal{M}$. Given a vector bundle $E\rightarrow M$ over the space-time manifold, we define the vector bundle $\mathcal{E}\rightarrow \mathcal{M}$ whose sections a smooth functionals of fields with values in $\Gamma(E)$. In other words, a section $\xi\in \Gamma(\mathcal{E})$ takes each field configuration $\varphi \in \mathcal{M}$ to a section $\xi[\varphi]\in \Gamma(E)$. Here we do not require the section $\xi[\varphi]$ to be smooth; discontinuous or even  distributional sections are also allowed. We will refer to $\mathcal{E}$ as the vector bundle associated with $E$. The dual vector bundle $\mathcal{E}^\ast$ is defined to be the vector bundle associated with $E^\ast$.

In order to justify our subsequent constructions some restrictions are to be imposed on the structure of the underlying space-time manifold. Our basic assumption will be that $M$ is a globally hyperbolic manifold endowed with a volume form. In the most of field-theoretical models both the structures come from a Lorentzian metric on $M$. The globally hyperbolic manifolds is a natural arena  for  the theory of hyperbolic differential equations with well-posed Cauchy problem. By definition, each  globally hyperbolic manifold  $M$ admits a global time function whose level surfaces provide a foliation of $M$ into space-like Cauchy surfaces $N$, so that $M\simeq \mathbb{R}\times N$. Using the direct product structure, one can cut $M$ into the ``past'' and the ``future'' with respect to a given instant of time $t\in \mathbb{R}$:
$$
M^-_t=(-\infty, t]\times N\,,\qquad M^+_t=[t,\infty )\times N\,,\qquad M=M_t^-\cup M_t^+\,.
$$

Given a vector bundle $E\rightarrow M$, we define the following subspaces in the space of sections $\Gamma(E)$:

\begin{itemize}
  \item $\Gamma_0(E)=\{\xi\in \Gamma(E)\,|\, \mbox{supp $\xi$ is compact}\}$;
  \item $\Gamma_{sc}(E)=\{\xi \in \Gamma(E)\,|\, \mbox{supp $\xi$ is spatially compact}\}$;
  \item $\Gamma_-(E)=\{\xi\in \Gamma_{sc}(E)\,|\, \mbox{supp $\xi \subset M_t^-$ for some $t$}\}$;
  \item $\Gamma_+(E)=\{\xi \in\Gamma_{sc}(E)\,|\,\mbox{supp $\xi\subset M^+_t$ for some $t$}\}$.
\end{itemize}
Here the spatially compact means that the intersection $M_t^-\cap \mathrm{supp} \,\xi \cap M_{t'}^+$ is compact for any $t\geq t'$.  We will refer to the elements of $\Gamma_-(E)$ and $\Gamma_+(E)$ as the sections with retarded and advanced support, respectively.

A differentiable functional $A$ is said to be compactly supported if $A,_{i}\in \Gamma_0(T^\ast\mathcal{M})$. For example, a local functional, like the action functional, is compactly supported if it is given by an integral over a compact  domain. The formally smooth and compactly supported functionals form a subalgebra $\Phi_0 \subset \Phi$.
Let now $\mathcal{E}$ be a vector bundle associated with $E$. We say that a section $\xi \in \Gamma(\mathcal{E})$ has retarded, advanced or compact support if $\xi[\varphi]\in \Gamma(E)$ does so for any field configuration $\varphi\in \mathcal{M}$. The sections with the mentioned support properties  form
subspaces in $\Gamma(\mathcal{E})$, which will be denoted by $\Gamma_-(\mathcal{E})$, $\Gamma_+(\mathcal{E})$, and $\Gamma_0(\mathcal{E})$, respectively.

When dealing with local field theories it is also useful to introduce the subspace of local sections $\Gamma_{loc}(\mathcal{E})\subset \Gamma(\mathcal{E})$. This consists of those   sections of $E$ whose components are given, in each coordinate chart, by smooth functions of the field $\varphi$ and its partial derivatives up to some finite order. For instance, the components of the Euler-Lagrange equations $S,_i=0$ constitute a section of $\Gamma_{loc}(T^\ast \mathcal{\mathcal{\mathcal{M}}})$.

\subsection{Dynamics}The dynamics of fields are specified by a set of differential equations
\begin{equation}\label{T}
T_a[\varphi]=0\,.
\end{equation}
Here $a$ is to be understood as including a space-time point. According to our definitions, the left hand sides of the equations can be viewed  as components of a local section of some vector bundle $\mathcal{E}$ over $\mathcal{M}$. We call $\mathcal{E}$ the \textit{dynamics bundle}. Since we do not assume the field equations (\ref{T}) to come from the least action principle, the discrete part of the condensed index $a$ may have nothing to do with that of $i$ labeling the field components. In the special case of Lagrangian systems, the dynamics bundle coincides with the cotangent bundle $T^\ast \mathcal{M}$ and the field equations  are determined by the exact 1-form  (\ref{dS}), with $S$ being the action functional.

Let $\Sigma$ denote the space of all solutions to the field equations (\ref{T}). Geometrically, we can think of $\Sigma$ as a smooth submanifold of $\mathcal{M}$ and refer to $\Sigma$ as the \textit{dynamical shell} or just the \textit{shell}. For the Lagrangian systems the shell is just the set of all stationary points of the action $S$. By referring  to  $\Sigma$ as a smooth submanifold we mean that the standard regularity conditions hold for the field equations \cite{H}.

 Given the shell, a functional $A\in \Phi_0$ is said to be trivial iff $A|_\Sigma=0$. Clearly, the trivial functionals form an ideal of the algebra $\Phi_0$. Denoting this ideal by  $\Phi_0^{\mathrm{triv}}$, we define the quotient-algebra $\Phi_0^\Sigma=\Phi_0/\Phi_0^{\mathrm{triv}}$. The regularity of the field equations  imply that for each trivial functional $A\in \Phi_0^{\mathrm{triv}}$  there exists a (distributional) section $\xi\in \Gamma(\mathcal{E}^\ast)$ such that $A=\xi^a T_a$.
 In other words, the trivial functionals are precisely those that are proportional to the equations of motion and their differential consequences. By definition, the elements of the algebra  $\Phi_0^\Sigma$ are given by the equivalence classes of functionals from $\Phi_0$, where two functionals $A$ and $B$ are considered to be equivalent if $A-B\in \Phi_0^{\mathrm{triv}}$. In that case we will write $A\approx B$. Formally, one can think of $\Phi_0^\Sigma$ as the space of smooth, compactly supported functionals on $\Sigma$.

\subsection{Gauge symmetries and physical observables}
The first functional derivatives of the field equations (\ref{T}) constitute the Jacobi operator $J$,
\begin{equation}\label{J}
J_{ai}=T_{a,i}\,.
\end{equation}
Geometrically, $J$ defines a homomorphism  from the tangent to the dynamics bundle.

The field equations (\ref{T}) are said to be \textit{gauge invariant} if there exist a vector bundle $\mathcal{F}\rightarrow \mathcal{M}$ together with a section $R=\{R_\alpha^i\}$ of $\mathcal{F}^\ast\otimes T\mathcal{M}$ such that
\begin{equation}\label{JR}
J_{ai}R^i_\alpha \approx 0\quad \Leftrightarrow \quad R_\alpha T_a=C_{\alpha a}^bT_b\,.
\end{equation}
In local field theory it is also assumed that $R_\alpha^i[\varphi]$ is the integral kernel of  a differential operator $R[\varphi]: \Gamma({\mathcal{F}})\rightarrow \Gamma(T{M})$ for each $\varphi \in \mathcal{M}$.

For the sake of simplicity we assume the bundle $\mathcal{F}$ to be trivial and consider  $R_\alpha=\{R_\alpha^i\}$ as a collection of vector fields on $\mathcal{M}$.  This vector fields are called the gauge symmetry generators. The terminology is justified by the fact that for any infinitesimal section $\varepsilon\in \Gamma_0(\mathcal{F})$ the infinitesimal change of field $\varphi^i\rightarrow \varphi^i +\delta_\varepsilon \varphi^i$, where $$
\delta_{\varepsilon}\varphi^i=R^i_\alpha\varepsilon^\alpha\,,
$$
 maps solutions of (\ref{T}) to solutions.  In other words, the vector fields $R_\alpha$ are tangent to the dynamical shell $\Sigma$. The gauge symmetry transformations are said to be trivial if $R\approx 0$.
 If the vector bundle $\mathcal{F}$ is big enough to accommodate all nontrivial gauge symmetries, then we call $\mathcal{F}$ the \textit{gauge algebra bundle} and refer to $R_\alpha$ as a complete set of gauge symmetry generators.  It follows from the definition (\ref{JR}) that the vector fields $R_\alpha$ define an on-shell involutive vector distribution on $\mathcal{M}$, i.e.,
 $$
 [R_\alpha, R_\beta]\approx C_{\alpha\beta}^\gamma R_\gamma\,,
 $$
 for some $C$'s. This distribution will be denoted by $\mathcal{R}$ and called \textit{gauge distribution}.

A functional $A\in \Phi_0$ is gauge invariant if
$$
A,_iR^i_\alpha\approx 0\,.
$$
In that case we say that $A$ represents a \textit{physical observable}. The gauge invariant functionals form a subalgebra $\Phi_0^{\mathrm{inv}}\subset \Phi_0$. Two gauge invariant functionals $A$ and $A'$ are considered as equivalent or represent the same physical observable if $A\approx A'$. So, we identify the physical observables
with the equivalence classes of gauge invariant functionals from $\Phi_0$. This definition is consistent as the trivial functionals are automatically gauge invariant and the property of being gauge invariant passes through the quotient $\Phi^{\mathrm{inv}}_0/\Phi_0^{\mathrm{triv}}$.  In what follows we will identify physical observables with their particular representatives in $\Phi^{\mathrm{inv}}_0$.

\section{The Lagrange anchor}\label{2}

 According to our definitions each classical field theory is completely specified by a pair  $(\mathcal{E}, T)$, where $\mathcal{E}\rightarrow \mathcal{M}$ is a vector bundle over the configuration space of fields and $T$ is a particular section of $\Gamma_{loc}(\mathcal{E})$. The solution space   $\Sigma$ is then identified with the zero locus of the section $T$.
 Whereas the classical equations of motion $T_a[\varphi]=0$ are enough to formulate the classical dynamics they are certainly insufficient for constructing a quantum-mechanical description of fields. Any quantization procedure has to involve one or another additional structure.  Within the path-integral
quantization, for instance, it is the action functional that plays the role of such an extra
structure. The procedure of canonical quantization relies on the Hamiltonian form of dynamics,
involving a non-degenerate Poisson bracket and a Hamiltonian. Either approach assumes the existence of a variational formulation for the classical equations of motion (the least action principle)
and becomes inapplicable beyond the scope of variational dynamics. The extension of these quantization methods to non-variational dynamics was proposed in \cite{KazLS}, \cite{LS0}.
In particular, the least action principle of the Lagrangian formalism was shown to admit a far-reaching  generalization based  on the concept of a \textit{Lagrange anchor}.

Like many fundamental concepts, the notion of a Lagrange anchor can be introduced and motivated from various  perspectives. Some of these motivations and interpretations can be found in Refs. \cite{KazLS}, \cite{LS1}, \cite{KLS1}.
For our present purposes it is convenient to define the Lagrange anchor $V$ as a linear operator making the on-shell commutative diagram
\begin{equation}\label{LD}
\xymatrix@C=0.5cm{
   \Gamma(T\mathcal{M}) \ar[rr]^{J} && \Gamma(\mathcal{E}) \\
  \Gamma(\mathcal{E}^\ast) \ar[u]_{V}\ar[rr]^{J^\ast} && \Gamma(T^\ast\mathcal{M})\ar[u]_{V^\ast}\\
  }
\end{equation}
Here $J^\ast$ and $V^\ast$ denote the formal dual of the operators $J$ and $V$.
The on-shell commutativity of the diagram means that
\begin{equation}\label{LAD}
J\circ V\approx V^\ast\circ J^\ast\,.
\end{equation}
Due to the regularity condition, the off-shell form of the last equality reads
\begin{equation}\label{LAD1}
J_{ai}V^i_b-V_a^iJ_{bi}=C_{ab}^d T_d
\end{equation}
for some $C$'s.

 In the case of Lagranian theories $\mathcal{E}=T^\ast \mathcal{M}$ and  we can take $V=1$.   Then   (\ref{LAD1}) reduces to the commutativity of the second functional derivatives, $J_{ij}=S,_{ij}=J_{ji}$. The operator $V=1$ is referred to as the \textit{canonical Lagrange anchor} for Lagrangian equations of motion. It should be noted that even for Lagrangian equations $S,_i=0$ there may exist non-canonical Lagrange anchors.

As with the generators of gauge symmetries, we can think of the Lagrange anchor as a collection of vector fields $V_a=\{V_a^i\}$ on $\mathcal{M}$. These generate a (singular) vector distribution $\mathcal{V}$, which we call the \textit{anchor distribution}. From the physical standpoint,  $\mathcal{V}$ defines the possible directions of  quantum fluctuations on $\mathcal{M}$. For the Lagrangian theories endowed with the canonical Lagrange anchor $V=1$ all directions are allowable and equivalent. At the other extreme we have zero Lagrange anchor, $V=0$, for which the corresponding quantum system remains pure classical (no quantum fluctuations). In the intermediate situation only a part of physical degrees of freedom may fluctuate and/or the intensity of  fluctuations around a particular field configuration $\varphi\in \mathcal{M}$ may vary with $\varphi$.

Unlike the gauge distribution $\mathcal{R}$, the anchor distribution $\mathcal{V}$ is not generally involutive even on shell. Nonetheless, using the regularity condition, one can derive the following commutation relations \cite{Sh}:
\begin{equation}\label{L1,2}
\begin{array}{l}
[V_a,V_b]^i\approx C_{ab}^dV^i_d+D_{ab}^\alpha R^i_\alpha +J_{aj}W_b^{ji}-J_{bj}W_a^{ji}\,,\\[3mm]
[R_\alpha, V_a]^i \approx C_{\alpha a}^bV^i_b+D_{\alpha a}^\beta R^i_\beta+J_{aj}W^{ji}_\alpha\,,
\end{array}
\end{equation}
where the coefficients  $W$'s are symmetric in $ij$ and the $C$'s are defined by Rels. (\ref{JR}) and (\ref{LAD1}). By definition, the coefficients $C_{ab}^d$ and $C_{\alpha a}^b$ are given by the integral kernels of trilinear differential operators, while the coefficients  $D$'s and $W$'s may well be nonlocal. Locality of the latter coefficients will be our additional assumption. It is automatically satisfied for the so-called \textit{integrable Lagrange anchors} as they were defined in \cite{KLS2}.  We will not dwell here on the concept of integrability of the Lagrange anchors referring the interested reader to the cited  paper.

\section{Covariant Poisson brackets}

The cornerstone of our construction is an \textit{advanced/retarded  fluctuation} $V_A^{\pm}$ caused by a physical observable $A$. By definition, $V_A^{\pm}$ is a vector field from $\Gamma_{\pm} (T\mathcal{M})$ satisfying the condition
\begin{equation}\label{fluct}
V_A^{\pm} T_a\approx V_a A\,.
\end{equation}
   It is not hard to see that the last equation  defines $V_A^{\pm}$ uniquely up to adding a vector field from $\mathcal{R}$ and on-shell vanishing terms \cite{Sh}.

Now we define the advanced/retarded brackets of two physical observables by the relation
\begin{equation}\label{PB}
\{A,B\}^{\pm}=V^{\pm}_AB-V^{\pm}_BA\,, \qquad \forall A, B\in \Phi^{\mathrm{inv}}_0\,.
\end{equation}
These brackets are well defined on shell as the ambiguity related to  the choice of the fluctuations,
\begin{equation}
V_A^\pm\quad \rightarrow\quad V_A^\pm+\xi^\alpha R_\alpha+T_a X^a \,,\qquad\xi\in \Gamma_\pm(\mathcal{F})\,,\quad X^a\in \Gamma_{\pm}(T\mathcal{M})\,,
\end{equation}
results in on-shell vanishing terms. Using Rels. (\ref{L1,2}) one can prove the following main statement.

\begin{prop}
{Brackets (\ref{PB}) map physical observables to physical observables and satisfy all the properties of Poisson brackets: antisymmetry, bi-linearity, the Leibnitz rule and Jacobi identity.}
\end{prop}
\begin{proof} See \cite{Sh}.
\end{proof}

In \cite{Sh}, it was also shown that the advanced and retarded fluctuations are connected to each other by the following \textit{reciprocity relation}:
\begin{equation}\label{RecRel}
V^-_AB\approx V^+_BA\,.
\end{equation}
It just says that the retarded effect of $A$ on $B$ is equal to the advanced effect of $B$
on $A$, and vice versa. As an immediate consequence of (\ref{RecRel}) we obtain the on-shell equality
\begin{equation}\label{PB1}
\{A,B\}^\pm\approx \pm\tilde{V}_AB\,,
\end{equation}
where the difference $\tilde{V}_A=V^+_A-V^-_A$ is called the \textit{causal fluctuation}.

Let us now compare the covariant Poisson brackets (\ref{PB}) with the usual Peierls' brackets in  Lagrangian field theory. In the latter case the dynamics of fields are governed by some action functional $S[\varphi]$. As was explained in Sec. 3, the corresponding equations of motion $S,_i[\varphi]=0$ admit the canonical Lagrange anchor given by the unit operator $V=1$ on $\Gamma(T^\ast \mathcal{M})$. The definition of the advanced/retarded fluctuation (\ref{fluct}) takes the form
\begin{equation}\label{V=1}
V_A^{\pm i}S,_{ij}\approx A,_j\,.
\end{equation}
In the absence of gauge symmetries this equation can be solved for $V_A^{\pm}$ with the help of the advanced/retarded Green function $G^{\pm ij}$. By definition,
\begin{equation}\label{GF}
G^{\pm in}S,_{nj}=S,_{jn}G^{\pm ni}=\delta^i_j\quad \mbox{and}
\quad G^{- ij}= 0=G^{+ji} \quad \mbox{if}\quad j>i\,.
\end{equation}
Here $j>i$ means that the time associated with the index $i$ lies to the past of the time associated with the index $j$. Besides (\ref{GF}),  the advanced and retarded Green functions satisfy the so-called reciprocity relation $G^{\pm ij}= G^{\mp ji}$.
In terms of the Green functions  the advanced/retared solution to (\ref{V=1}) reads
\begin{equation}\label{VGA}
  V_A^{\pm i}=G^{\pm ij}A,_j\,.
\end{equation}
and the causal fluctuation takes the form $\tilde{V}^i_A=V^{+i}_A-V^{-i}_A=\tilde{G}^{ij}A,_j$, where the difference $\tilde{G}=G^+-G^-$ is known as the causal Green function.  In view of the reciprocity relation, $\tilde{G}^{ij}=-\tilde{G}^{ji}$.  Substituting  (\ref{VGA}) into (\ref{PB1}), we get
\begin{equation}\label{AGB}
  \{A,B\}^\pm=\pm A,_i\tilde{G}^{ij}B,_j\,.
\end{equation}
The antisymmetry of the brackets as well as the derivation property are obvious. The direct verification of the Jacobi identity for (\ref{AGB}) can be found in \cite{DWbook}.

\section{An example: the Pais-Uhlenbeck oscillator}

The  PU oscillator is described by the fourth-order differential equation
\begin{equation}\label{PU}
\left(\frac{d^2}{dt^2}+\omega^2_1\right)\left(\frac{d^2}{dt^2}+\omega_2^2\right)x=0\,,
\end{equation}
where the constants $\omega_{1}$ and $\omega_2$ have the meaning of frequencies.
The advanced/retarded Green function $G^\pm (t_2-t_1)$ for (\ref{PU}) is given by
$$
G^\pm(t)=\pm \frac{\theta(\mp t)}{\omega_2^2-\omega_1^2}\left(\frac{\sin \omega_1t}{\omega_1}-\frac{\sin\omega_2t}{\omega_2}\right)\,,
$$
with $\theta(t)$ being the Heaviside step function.

Equation (\ref{PU}) admits the two-parameter family of the Lagrange anchors \cite{KLSh}
\begin{equation}\label{VPU}
V= \alpha+\beta \frac{d^2}{dt^2}\,,\qquad \alpha,\beta \in \mathbb{R}\,.
\end{equation}
In this particular case the defining condition for the Lagrange anchor (\ref{LAD1}) reduces to the commutativity of the operator $V$ with the fourth-order differential operator defining the equation of motion  (\ref{PU}). Notice that the equation of motion (\ref{PU}) is Lagrangian and the canonical Lagrange anchor corresponds to $\alpha=1$, $\beta=0$.

The advanced Poisson brackets are given by
$$
\{x(t_1), x(t_2)\}^+= V\tilde{G}(t_1-t_2)
$$
$$=\left(\frac{\alpha-\beta\omega^2_1}{{\omega_2^2-\omega_1^2}}\right)\frac{\sin\omega_1(t_1-t_2)}{\omega_1}
-\left(\frac{\alpha-\beta\omega^2_2}{{\omega_2^2-\omega_1^2}}\right)\frac{\sin\omega_2(t_1-t_2)}{\omega_2}\,.$$

Differentiating by $t_1$, $t_2$ and setting $t_1=t_2$, we obtain the following  Poisson brackets of the phase-space variables $z=(x, \dot x, \ddot x, \dddot x)$:
\begin{equation}\label{PU-PB}
\begin{array}{c}
\{\dot x, x\}^+=\beta\,,\qquad \{\dot x, \ddot x\}^+=   \{\dddot x, x\}^+=\alpha-\beta(\omega_1^2+\omega_2^2)\,,\\[3mm] \{\ddot x, \dddot x\}^+=\alpha(\omega_1^2+\omega_2^2)-\beta(\omega_1^4+\omega_1^2\omega_2^2+\omega_2^4)\,,
\end{array}
\end{equation}
and the other brackets vanish. For $\alpha=1$, $\beta=0$, this yields  the standard Poisson brackets on the phase space of the PU oscillator.

With the Poisson brackets (\ref{PU-PB}) the equations of motion (\ref{PU}) can be written in the Hamiltonian form
$$
\dot z^i=\{H, z^i\}^+\,, \qquad i=1, 2, 3, 4\,,
$$
where the Hamiltonian is given by
$$
H=\frac12\frac{(\dddot x+\omega_1^2\dot x)^2+\omega_2^2(\ddot x+\omega_1^2 x)^2}{(\omega_1^2-\omega_2^2)(\alpha-\beta\omega^2_2)}-
\frac12\frac{(\dddot x+\omega_2^2\dot x)^2+\omega_1^2(\ddot x+\omega_2^2 x)^2}{(\omega_1^2-\omega_2^2)(\alpha-\beta\omega^2_1)}\,.
$$
As was first noticed in \cite{KLSh}, this Hamiltonian is positive definite whenever
\begin{equation}\label{ineq}
\omega_1^2>\frac{\alpha}{\beta}>\omega_2^2\,.
\end{equation}
Clearly, the canonical Lagrange anchor $(\alpha=1, \beta=0)$ does not satisfy these inequalities for any frequencies $\omega_{1,2}$. On the other hand, in the absence of resonance $(\omega_1\neq \omega_2)$, one can always choose the non-canonical Lagrange anchor (\ref{VPU}) to meet inequalities (\ref{ineq}). Upon quantization the positive-definite Hamiltonian will have a positive energy spectrum and a well-defined ground state.  The last property is crucial for the quantum stability of the system \cite{KLSh}. So, we see that non-canonical Lagrange anchors may offer  certain advantages over the canonical one, when the issuers of quantum stability of higher-derivative systems are concerned.

\subsection*{Acknowledgment}
I wish to thank Simon Lyakhovich for a fruitful collaboration at the early stage of this work.

\end{document}